\documentclass[12pt]{article}

\usepackage[T1]{fontenc}
\usepackage[utf8]{inputenc}
\usepackage{lmodern}

\usepackage{amsmath,amssymb,amsthm,amsfonts,mathrsfs,latexsym}

\usepackage{enumerate}
\usepackage{ragged2e}
\usepackage{graphicx}
\usepackage{caption}
\usepackage{subcaption}
\usepackage{float}

\usepackage{tikz}
\usetikzlibrary{intersections,arrows}
\usepackage{tkz-tab}
\usepackage{pgfplots}
\pgfplotsset{compat=1.15}

\usepackage[colorlinks=true,allcolors=blue]{hyperref}
\usepackage[all]{hypcap}
\usepackage[capitalize,nameinlink]{cleveref}

\usepackage{titlesec}
\titlelabel{\thetitle.\quad}

\newtheorem{theo}{Theorem}

\theoremstyle{definition}

\begin{document}
\title{On powers of circular arc graphs :  A Note}
\def\correspondingauthor{\footnote{Corresponding author}}
\author{Ashok Kumar Das\correspondingauthor{}  and  Indrajit Paul  \\Department of Pure Mathematics\\
University of Calcutta\\
35, Ballygunge Circular Road\\
Kolkata-700019\\
Email Address - ashokdas.cu@gmail.com \&
paulindrajit199822@gmail.com}
\date{}
\maketitle
\begin{abstract}
   A class of graphs $\mathcal{C}$ is closed under powers if for every graph $G\in\mathcal{C}$ and every $k\in\mathbb{N}$, $G^k\in\mathcal{C}$. Also $\mathcal{C}$ is strongly closed under powers if for every $k\in\mathbb{N}$, if $G^k\in\mathcal{C}$, then $G^{k+1}\in\mathcal{C}$. It is known that circular arc graphs and proper circular arc graphs are closed under powers. But it is open whether these classes of graphs are also strongly closed under powers. In this note we have settled these problems. Our method also gives a simple proof that the class of  circular arc graphs and proper circular arc graphs are closed under powers.
\end{abstract}
\section{Introduction}
Let $G$ be a simple graph, finite and connected graph. The \textit{distance} between two vertices $u, v\in V(G)$ in $G$, denoted by $d_G(u,v)$ is the length of the shortest $u$-$v$ path. The \textit{diameter} of $G$, denoted by \textit{diam}$(G)$, is the largest distance between any pair of vertices in $G$. For a positive integer $k$, the $k$th power of a graph $G=(V,E)$, denoted by $G^k$, is a graph whose vertex set is $V$ and two vertices $u$ and $v$ are adjacent in $G^k$ if and only if $d_G(u,v)\leq k$.\\
\par A graph class $\mathcal{C}$ is \textit{closed under powers} if for every $G\in\mathcal{C}$ and for any natural number $k$, $G^k\in\mathcal{C}$. And $\mathcal{C}$ is \textit{strongly closed under powers} if for every $k\in\mathbb{N}$, if $G^k\in \mathcal{C}$ then also $G^{k+1}\in\mathcal{C}$. However, powers of graphs have been studied from different view points, it is interesting to the researchers to investigate which classes of graphs are closed or strongly closed under powers. The first result in this direction was given by Lubiw \cite{6}, who proved that all the powers of strongly chordal graphs (for definition see \cite{mcg}) are strongly chordal. This result has been improved by Raychoudhuri \cite{9}, showing that the class of strongly chordal graphs is strongly closed under powers. Now many results in this line are known. Chang and Nemhauser \cite{1} proved that for a chordal graph $G$ if $G^2$ is chordal then all the powers of $G$ are also chordal. A more crucial result with elegant proof for this class of graphs was given by Duchet \cite{3}. He proved that, if $G^k$ is chordal, then so is $G^{k+2}$. A systematic investigation of some classes of graphs closed and strongly closed under powers has been done by Flotow \cite{4}. In this paper we shall prove some open problems in this area.

\section{Preliminaries}
A graph $G=(V,E)$ is an \textit{interval graph}, if corresponding to each vertex we can assign an interval of the real line such that two vertices of $G$ are adjacent if and only if their corresponding intervals intersect. An interval graph $G$ is a \textit{proper interval graph}, if there is an interval representation of $G$ such that no interval is properly contained in another.\\
A graph $G=(V,E)$ is a \textit{circular arc graph }, if corresponding to each vertex of $V$ we can assign a circular arc of a host circle such that two vertices of $G$ are adjacent if and only if their corresponding arcs intersect. A circular arc graph $G$ is a \textit{proper circular graph} if there exists a circular arc representation such that no arc is properly contained in another.\\
Raychaudhuri \cite{8} proved that the class of interval graphs is strongly closed under powers. Ramalingam and Pandu Rangan \cite{7} characterized an interval graph in terms of its vertex ordering, which is known as interval ordering, defined as follows.\\
An \textit{interval ordering} of $G$ is an ordering of $V(G)$ into $[v_1,v_2,...,v_n]$ such that $i<l<j$ and $v_iv_j\in E(G)$ imply $v_lv_j\in E(G)$. Ramalingam and Pandu Rangan proved that a graph is an interval graph if and only if it has an interval ordering. M.Chen and G.J.Chang \cite{2} gave a much simple proof of Raychaudhuri's result using the vertex ordering characterization of interval graphs. In the same paper Raychaudhuri also proved that the class of proper interval graphs is strongly closed under powers. It can be observed that in the interval representation of a proper interval graph ordering of the left end points of the intervals is same as the ordering of the right end points of the intervals. With this observation in mind Roberts \cite{10} defined proper interval ordering of the vertex set $V(G)$ of $G$ in the following way. If the vertex set has a ordering $[v_1,v_2,...,v_n]$ such that $[v_1,v_2,...,v_n]$ and $[v_n,v_{n-1},...,v_2,v_1]$ are interval orderings, equivalently $i<l<j$ and $v_iv_j\in E(G)$ imply $v_iv_l\in E(G)$ and $v_lv_j\in E(G)$, then the ordering of the vertex set $V(G)$ is a \textit{proper interval ordering}. Roberts \cite{10} also proved that $G$ is a proper interval graph if and only if $G$ has a proper interval ordering. Using this characterization Chen and Chang \cite{2} gave a simple proof of Raychaudhuri's \cite{8} result that the class of proper interval graphs is strongly closed under powers.\\
For the class of circular arc graphs, Raychaudhuri \cite{9} proved that this class is also closed under powers. Flotow \cite{5} gave an independent proof of this result.\\
Flotow \cite{5} also proved the following.
\vspace{.4cm}
\par\noindent\textbf{Theorem A.}\hspace{.2cm} \textit{If $G$ is a proper circular arc graph, then $G^k$ is a proper circular arc graph.}\vspace{.3cm}\\
But their proofs are much complicated and need other results. Also it is open whether the class of  circular arc graphs is strongly closed under powers. Following progress in this direction was made by Flotow \cite{5}.\\ \vspace{.3cm}
\par\noindent\textbf{Theorem B.} \textit{For any graph $G$}
\begin{enumerate}[(i)]
    \item  \textit{if $G^k$ is a circular arc graph, then so is $G^{k+2}$,}
    \item \textit{if \textit{diam}$(G^k)\geq 4$ and $G^k$ is a circular arc graph, then so is $G^{k+1}$.}
\end{enumerate}

Also it is open whether the class of proper circular arc graphs is strongly closed under powers. In this paper we have given  simple proofs of these open problems.
\section{Main Result}
Tucker \cite{11} characterized circular arc graphs and proper circular arc graphs in terms of their augmented adjacency matrices. Motivated by the interval ordering characterization \cite{7} of interval graphs, we give here an ordering characterization of circular arc graph. Let $G=(V,E)$ be an undirected graph. We index its vertices as $v_1,v_2,...,v_n$ such that if we place the vertices $v_1,v_2,...,v_n$ on a circle $C$, they will appear circularly on $C$ in the clockwise movement. Now, we say $\sigma$=$[v_1,v_2,...v_n]$ is a \textit{circular ordering} if $v_iv_j\in E(G)$ and $i<j$ then $v_iv_{i+1}, v_iv_{i+2},...,v_iv_{j-1}\in E$ or $v_{i-1}v_j,v_{i-2}v_j,...,v_1v_j,v_nv_j,...v_{j+1}v_j\in E$.\\ 
Now, we characterize a circular arc graph in terms of circular ordering of the vertices of the graph.\\
\begin{theo}\label{t1}
A graph $G=(V,E)$ is a circular arc graph if and only if the vertex set $V(G)$ has a circular ordering.
\end{theo}
\begin{proof}
Let $G=(V,E)$ be a circular arc graph. So $G$ has a circular arc representation (with closed arcs of an $n$-hour clock). Now, order the vertices of $G$ according to increasing order of counterclockwise end points of the corresponding arcs moving clockwise around the circle. Let $\sigma=[v_1,v_2,...,v_n]$ be the such ordering. Let $A_i$ be the arc corresponding to $v_i$. Now assume $A_i\cap A_j\neq\emptyset$, where $i<j$. First, assume that the counterclockwise end point of $A_j$ belongs to $A_i$. Then for all $k$ lies between $i$ and $j$ in the clockwise movement, the counterclockwise end point of $A_k$ belongs to $A_i$. Thus $v_iv_j\in E(G)$ implies $v_iv_{i+1}, v_iv_{i+2},\dots, v_iv_{j-1}\in E(G)$. In the other case, assume counterclockwise end point of $A_i$ belongs to $A_j$. Then for all $k$ lies between $j$ and $i$ in the clockwise movement, the counterclockwise end point of $A_k$ lies in $A_j$. Hence $v_iv_j\in E(G)$ implies $v_{i-1}v_j, v_{i-2}v_j,\dots, v_1v_j,v_nv_j,\dots, v_{j+1}v_j\in E(G)$. Hence the vertex set $V(G)$ has a circular ordering.  
\par For the converse, let the vertex set $V(G)$ has a circular ordering $\sigma=[v_1,v_2,...,v_n]$, where the vertices $v_1,v_2,...,v_n$ are placed  in the clockwise pattern 
on an $n$-hour clock. Also corresponding to the vertex $v_i$, we have the $i$th hour point on the clock. Now, we shall construct the circular arc $A_i$ corresponding to the vertex $v_i$ as follows. $A_i=[i,a]$ where $v_a$ is the last vertex adjacent to $v_i$ in the circular ordering $\sigma$. Then $A_i\cap A_j\neq\emptyset$ $(i<j)$ if and only if counterclockwise end point of $A_j$ lies in $A_i$ or counterclockwise end point of $A_i$ lies in $A_j$. Thus $v_iv_j\in E$ if and only if $A_i\cap A_j\neq\emptyset$.  This completes the proof.
\end{proof}
As stated before, it is known that the class of circular arc graphs is closed under powers. In the next theorem we shall prove that it is also strongly closed under powers using the above theorem.
\begin{theo}\label{t2}
Let $G$ be any graph and $k$ be any natural number. If $G^k$ is a circular arc graph, then $G^{k+1}$ is also a circular arc graph.
\end{theo}
\begin{proof}
Since $G^k$ is a circular-arc graph, it has a circular arc representation. In this representation label the vertices according to the increasing order of their counterclockwise end points of the arcs moving clockwise direction around the circle. Let $\sigma=[v_1,v_2,...,v_n]$ be the such ordering. By Theorem 1, it is a circular ordering of $G^k$. Next consider the graph $G^{k+1}$. We shall show that $G^{k+1}$ has a circular arc representation and $\sigma=[v_1,v_2,...,v_n]$ is also a circular ordering of $G^{k+1}$. Now assume $v_iv_j\in E(G^k)$, where $i<j$. Then also $v_iv_j\in E(G^{k+1})$. Thus in $G^{k+1}$, $v_i$ is adjacent to $v_{i+1}, v_{i+2},\dots, v_{j-1}$ or $v_j$ is adjacent to $v_{i-1},v_{i-2}...,v_1,v_n,...,v_{j+1}$; implies $\sigma$ is a circular ordering of $G^{k+1}$. Next assume $v_iv_j\in E(G^{k+1})$ but $v_iv_j\notin E(G^k)$. So, $d_G(v_i,v_j)=k+1$ and $P:v_i,v_p,\dots,v_r,v_j$ be a smallest path of length $k+1$ in $G$. Now the subpath $v_i$-$v_r$ of $P$ is a shortest path of length $k$ in $G$ i.e., $d_G(v_i,v_r)=k$. Again $v_rv_j\in E(G)$. First, we assume that counterclockwise end point of $A_j$ belongs to $A_r$. Again $v_iv_r\in E(G^k)$, now extend the clockwise end point of $A_i$ up to the counterclockwise end point of $A_j$. Suppose in this process extended arc $A_i$ intersects a new arc $A_s$. Obviously $A_r\cap A_s\neq\emptyset$, i.e., $v_rv_s\in E(G)$. Thus $v_iv_s\in E(G^{k+1})$. Thus doing this extension of arc we have not created any edge that is not in $E(G^{k+1})$ (see Figure 1).\\
\begin{figure}[H]
    \centering
   \begin{tikzpicture}[line cap=round,line join=round,>=triangle 45,x=1cm,y=1cm,scale=1.5]
\clip(6,4) rectangle (11,9);
\draw [line width=1.5pt] (8,7) circle (1.4142135623730951cm);
\draw [shift={(8,7)},line width=1pt]  plot[domain=3.119496802705948:4.706771061959734,variable=\t]({1*1.8104419349981922*cos(\t r)+0*1.8104419349981922*sin(\t r)},{0*1.8104419349981922*cos(\t r)+1*1.8104419349981922*sin(\t r)});
\draw [shift={(8,7)},line width=1pt]  plot[domain=3.9238362533187647:5.500941707450615,variable=\t]({1*2.2415396494374127*cos(\t r)+0*2.2415396494374127*sin(\t r)},{0*2.2415396494374127*cos(\t r)+1*2.2415396494374127*sin(\t r)});
\draw [shift={(8,7)},line width=1pt]  plot[domain=5.138670811582712:5.8589908992758195,variable=\t]({1*2.3942639787625755*cos(\t r)+0*2.3942639787625755*sin(\t r)},{0*2.3942639787625755*cos(\t r)+1*2.3942639787625755*sin(\t r)});
\draw [shift={(8,7)},line width=1pt]  plot[domain=4.794347830323393:5.003150606396158,variable=\t]({1*2.8094305472817798*cos(\t r)+0*2.8094305472817798*sin(\t r)},{0*2.8094305472817798*cos(\t r)+1*2.8094305472817798*sin(\t r)});
\begin{scriptsize}
\draw [fill=black] (8.55,4.64) circle (.5pt);
\draw [fill=black] (8.78,4.7) circle (.5pt);
\draw [fill=black] (8.01,4.62) circle (.5pt);
\draw [fill=black] (8.3,4.6) circle (.5pt);
\draw (9.8,5.5) node[anchor=north west,scale=1.5] {$A_i$};
\draw (7.9,4.65) node[anchor=north west,scale=1.5] {$A_s$};
\draw (7.5,5.2) node[anchor=north west,scale=1.5] {$A_r$};
\draw (6.25,6.5) node[anchor=north west,scale=1.5] {$A_j$};

\end{scriptsize}
\end{tikzpicture}
    \caption{Extension of the clockwise end-point of $A_i$ up to anticlockwise end-point of $A_j$.}
    \label{fig:placeholder}
\end{figure}
Next, suppose the clockwise end point of $A_j$ belongs to $A_r$. Also $A_i\cap A_r\neq\emptyset$ as $v_iv_r\in E(G^k)$ and $G^k$ is a circular-arc graph.\\
Now, as before we extend the clockwise end point of $A_j$ up to the counterclockwise end point of $A_i$. Suppose in this process $A_j$ intersect a new arc $A_s$. Then $v_rv_s\in E(G)$ and $v_jv_s\in E(G^{k+1})$. Thus, again we conclude that doing this extension we have not created any edge that is not in $E(G^{k+1})$ (see Figure 2). 
\begin{figure}[H]
    \centering
    \begin{tikzpicture}[line cap=round,line join=round,>=triangle 45,x=1cm,y=1cm,scale=1.5]
\clip(5,4) rectangle (10,9);
\draw [line width=1.5pt] (8,7) circle (1.4142135623730951cm);
\draw [shift={(8,7)},line width=1.pt]  plot[domain=3.119496802705948:4.706771061959734,variable=\t]({1*1.8104419349981922*cos(\t r)+0*1.8104419349981922*sin(\t r)},{0*1.8104419349981922*cos(\t r)+1*1.8104419349981922*sin(\t r)});
\draw [shift={(8,7)},line width=1pt]  plot[domain=3.9238362533187647:5.500941707450615,variable=\t]({1*2.2415396494374127*cos(\t r)+0*2.2415396494374127*sin(\t r)},{0*2.2415396494374127*cos(\t r)+1*2.2415396494374127*sin(\t r)});
\draw [shift={(8,7)},line width=1pt]  plot[domain=2.072572586940641:3.117354976725056,variable=\t]({1*1.6424676556937121*cos(\t r)+0*1.6424676556937121*sin(\t r)},{0*1.6424676556937121*cos(\t r)+1*1.6424676556937121*sin(\t r)});
\draw [shift={(8,7)},line width=1pt]  plot[domain=3.3091606811227607:3.7461912436832696,variable=\t]({1*2.038553408669981*cos(\t r)+0*2.038553408669981*sin(\t r)},{0*2.038553408669981*cos(\t r)+1*2.038553408669981*sin(\t r)});
\begin{scriptsize}
\draw [fill=black] (5.79,7) circle (.5pt);
\draw [fill=black] (6.41,5.42) circle (.5pt);
\draw [fill=black] (5.812809375102383,6.683302714946952) circle (.5pt);
\draw [fill=black] (5.883477593804778,6.364049605650263) circle (.5pt);
\draw [fill=black] (6.0233159078901854,6.011657953945093) circle(.5pt);
\draw [fill=black] (6.186208096402429,5.737360332302244) circle (.5pt);
\draw (6.3,5.9) node[anchor=north west,scale=1.5] {$A_s$};
\draw (7.9,4.7) node[anchor=north west,scale=1.5] {$A_j$};
\draw (7.2,5.7) node[anchor=north west,scale=1.5] {$A_r$};
\draw (5.9,7.8) node[anchor=north west,scale=1.5] {$A_i$};

\end{scriptsize}
\end{tikzpicture}
    \caption{Extension of the clockwise end-point of $A_j$ up to the anticlockwise end-point of $A_i$.}
    \label{fig:placeholder}
\end{figure}
  This gives a circular arc representation of $G^{k+1}$, where the ordering of the counterclockwise end points is same as $\sigma$. This completes the proof of Theorem 2.   
\end{proof}
\textbf{Observation.} Using the method of above proof, it is easy to prove that the class of circular arc graphs is closed under powers.\vspace{.3cm}\\
Next, we shall give a circular ordering characterization of proper circular arc graphs. Let $G=(V,E)$ be a proper circular arc graph. Then in the circular arc representation $\{ A_v,v\in V\}$ of $G$, such that no arc is properly contained in another. Now, by \cref{t1}, the vertex set $V$ has a circular ordering. Let $\sigma=[v_1,v_2,...,v_n]$ be the such ordering. Also, as in the proof of \cref{t1}, assume corresponding to $v_i$ we have the $i$-th hour point on an $n$-hour clock. Next, assume corresponding to $v_i$ we have the arc $A_i$, where $A_i=[i,\hat{i}]$ and the arc is extending from counterclockwise endpoint $i$ to clockwise endpoint $\hat{i}$. Now $\sigma$ is the ordering of the vertices according to the increasing order of the counterclockwise endpoints of the arcs. Let $v_m$ be the first vertex in the ordering $\sigma$ such that $m>\hat{m}$. It may be noted that if there exists no such $v_m$, then the graph $G$ is a proper interval graph. Now, we  write $\sigma = [v_1,v_2,...,v_{m-1},v_m,v_{m+1},...,v_n]$. As $G$ is a proper circular arc graph, if $m<m_1<m_2$ then $\hat{m}<\hat{m_1}<\hat{m_2}$. Next, we order the vertices of $G$ according to the decreasing order of the clockwise endpoints of the arcs, then we have the ordering $\sigma'=[v_{m-1},v_{m-2},...,v_2,v_1,v_n,v_{n-1},...v_m]$ of the vertices. Now to show that $\sigma'$ is also a circular ordering, let $u_i$ be the vertex which is at the $i$th position in $\sigma'$. Then as in \cref{t1} we can easily prove that $\sigma'=[u_1,u_2,...,u_n]$ is a circular ordering. Thus $\sigma'=[v_{m-1},v_{m-2},...,v_2,v_1,v_n,v_{n-1},...v_m]$ is a circular ordering.\\

\par Now we say the vertex set $V(G)=\{v_1,v_2,...,v_n\}$ has a $\textit{proper circular ordering}$ if $V(G)$ has two orderings  $\sigma = [v_1,v_2,...,v_{m-1},v_m,v_{m+1},...,v_n]$ and $\sigma'=[v_{m-1},v_{m-2},...,v_2,v_1,v_n,v_{n-1}\\,...v_m]$ such that $\sigma$ and ${\sigma}'$ both are circular orderings. 
\begin{theo}\label{t3}
A graph $G=(V,E)$ is a proper circular arc graph if and only if its vertex set $V(G)$ has a proper circular ordering.

\end{theo}
\begin{proof}
Let $G$ be a proper circular arc graph. So it has a proper circular arc representation, where arcs are taken from an $n$-hour clock and $n$ is the order of $G$. Now as in preceding discussion the vertex set $V(G)$ has a proper circular ordering.\\
\par 
Conversly, let the vertex set $V(G)$ has a proper circular ordering. Thus $V(G)$ has two circular orderings $\sigma = [v_1,v_2,...,v_{m-1},v_m,v_{m+1},...,v_n]$ and $\sigma'=[v_{m-1},v_{m-2},...,v_2,v_1,v_n,v_{n-1}\\,...v_m]$. As in the \cref{t1}, we construct circular arcs. Let $A_i=[i,l]$, where $v_l$ is the last vertex adjacent to $v_i$ in the circular ordering $\sigma$. Also let $A_j=[j,l']$ where $i<j$ and $v_{l'}$ is the last vertex adjacent to $v_j$ in $\sigma$. Now $l'<l$ is not possible as $\sigma'$ is also a circular ordering. Next, if $l=l'$, then we replace $l'$ by $l'+1-\frac{1}{2j}$. Thus we have a proper circular arc representation of $G$.  This completes the proof. 
\end{proof}
\noindent Using the above results, we shall prove that the class of proper circular arc graphs is strongly closed under powers.

\begin{theo}\label{t4}
Let $G$ be any graph and $k$ be a positive integer. If $G^k$ is a proper circular arc graph, then so is $G^{k+1}$.
\end{theo}

\begin{proof}
Let $G^k$ be a proper circular arc graph and $[v_1,v_2,...,v_n]$ be a proper circular ordering of $G^k$. Thus both $[v_1,v_2,...,v_{m-1},v_m,v_{m+1},...,v_n]$ and $[v_{m-1},v_{m-2},...,v_2,v_1,v_n,v_{n-1},...v_m]$ are circular ordering of $G^k$. Now, by the arguments in the proof of \cref{t2}, we have that $[v_1,v_2,...,v_{m-1},v_m,v_{m+1},...,v_n]$ and $[v_{m-1},v_{m-2},...,v_2,v_1,v_n,v_{n-1},...v_m]$ are both circular orderings of $G^{k+1}$. Hence the vertices of $G^{k+1}$ has a proper circular ordering and so by \cref{t3}, $G^{k+1}$ is a proper circular arc graph.
\end{proof}\vspace{.5cm}
\textbf{Observation.} Using the method employed in the proof of the above theorem, we can easily show that the class of proper circular arc graphs is closed under powers.
\end{document}